\documentclass[18pt,fullpage,doublespace]{amsart}

\usepackage{amsmath, amsthm, amssymb}
\usepackage{multicol}

\usepackage{graphicx}

\newtheorem{thm}{Theorem}[section]
\newtheorem{cor}[thm]{Corollary}
\newtheorem{prop}[thm]{Proposition}

\newtheorem{rem}[thm]{Remark}
\newtheorem{example}[thm]{Example}

\numberwithin{equation}{section}

\def\D{\mathcal D}

\def\r{\text{\bf r}}

\newcommand{\rv}[1]{{\tt{#1}}}

\begin{document}

\author{Klaas Slooten}
\address{Netherlands Forensic Institute, P.O. Box 24044, 2490 AA The Hague, The Netherlands}
\email{k.slooten@nfi.minvenj.nl}
\author{Ronald Meester}
\address{VU University Amsterdam, De Boelelaan 1081, 1081 HV Amsterdam, The Netherlands}
\email{rmeester@few.vu.nl}
\title[Familial DNA Database Searching]{Forensic Identification: Database likelihood ratios and Familial DNA Searching}

\date{\today}
\begin{abstract}
Familial Searching is the process of searching in a DNA database for relatives of a certain individual. It is well known that in order to evaluate the genetic evidence in favour of a certain given form of relatedness between two individuals, one needs to calculate the appropriate likelihood ratio, which is in this context called a Kinship Index. Suppose that the database contains, for a given type of relative, at most one related individual. Given prior probabilities for being the relative for all persons in the database, we derive the likelihood ratio for each database member in favour of being that relative. This likelihood ratio takes all the Kinship Indices between the target individual and the members of the database into account. We also compute the corresponding posterior probabilities.
We then discuss two methods to select a subset from the database that contains the relative with a known probability, or at least a useful lower bound thereof. One method needs prior probabilities and yields posterior probabilities, the other does not. We discuss the relation between the approaches, and illustrate the methods with familial searching carried out in the Dutch National DNA Database.
\end{abstract}

\maketitle

\section{Introduction} Many countries maintain databases that contain forensic DNA profiles of traces and of certain known individuals, 
e.g.\ convicted offenders or suspects of certain crimes. These databases were originally set up to directly identify an unknown offender by looking for matching DNA profiles. However, since DNA is inhereted from parent to child, it is also possible to use them to look for the offender's relatives, rather than the offender himself, if the offender's DNA profile turns out not to be in the database. This last process is called familial (DNA) searching, and is carried out in several jurisdictions (e.g.\ the UK, some US states and New Zealand). As a result there have been some high profile successes (see e.g.\ \cite{miller} for the Grim Sleeper case). The Netherlands have recently adopted a law that allows familial searching in some cases.

Previous studies on familial DNA searching have mostly concentrated on empirical determination of the rank that a relative (of some target profile) in the database occupies, when the database is ordered according to decreasing likelihood ratio with the target, or according to decreasing number of shared alleles with the target. See e.g.\ \cite{bieber} for simulations that also includes a geographical component (based on US states), or \cite{cubu}, \cite{myers} and \cite{HicksII}. In \cite{ge}, false exclusion rates and false inclusion rates for various thresholds on the likelihood ratio and/or number of shared alleles are estimated by simulation. These rates are averages over different target profiles. 

The contributions of this paper to the research on familial searching are as follows. In part, we rediscover and collect several different statistical aspects of familial searching that are currently scattered in the literature. But moreover, we present and extend these results in a unified mathematical framework that allows for a comparison of the different approaches to familial searching. In particular we extensively discuss the different frequentist meaning of the probabilities involved. 

In this paper, more specifically, we define various search strategies that allow one to determine or control the probability with which the relative can be found, if indeed it is in the database. We also explain how these search strategies can deal with heterogeneous databases in which the amount of information stored for its members may differ between members (e.g., different sets of autosomal loci). It turns out that there are different strategies that are equally effective, and we discuss their different probabilistic interpretations. 

The first method, which we call the {\em conditional method}, gives a probabilistic model that allows one to obtain posterior probabilities for relatedness that take all the database information (and the prior probabilities) into account. Focus on the posterior can also be found in \cite{cavallini}, but indirectly as the result of Bayesian network computation, whereas we derive these probabilities without needing such a tool.

Our second method, which we call the {\em target-centered method}, selects a subset from the database that contains a relative with a certain probability. It essentially weighs false-negative (not selecting a true relative) probabilities against false positive (selecting an unrelated individual). The literature contains accounts of such approaches (e.g. \cite{ge}) but these studies focus on the average over all target profiles, whereas we focus on what happens in a particular case by taking the target profile as starting point.

We investigate the effectiveness of a familial search for a specific DNA target profile, rather than the average over all profiles. As is to be expected, if the target profile has more rare alleles, its relatives are easier to find then if it has more common alleles. This effect has been already noted (quantitavely) in \cite{CowenThomson} and (heuristically) in \cite{HicksII}. We illustrate the results by a comprehensive simulation study using artificial targets in the actual Dutch National DNA database. Finally, we look at how hard half-sibling would be to find either with the appropriate index (the half-sibling index) or with the sibling index. It turns ou that half-siblings are, of course, quite hard to find but that a search with sibling index does not perform very much worse than one with the half-sibling index.

In a familial search setting, we have a target item that we compare with database items among which, we suppose, at most one `related' item is present. If so, we wish to find it, and we do so by computing likelihood ratios in favour of being a special item between the target and every database member. 
From these, we compute (given prior probabilities) posterior probabilities for relatedness between the target and a database member that take {\it all} the computed one-to-one likelihood ratios into account. 

What distinguishes a familial search from an ordinary kinship computation is of course the fact that we compute likelihood ratios with a whole database. Databases are bound to yield chance matches, meaning in this case a strong indication for a non-existing relationship. This is reminiscent of the classical database controversy, and some elements of that discussion reappear in the familial searching context. For example, an attempt to take into account the fact that a database search has been done is described in \cite{SWGDAM} where the SWGDAM Ad Hoc Committee on Partial Matches recommends that the kinship index between a database member and the target be divided by $N$, the size of the database, and to only further investigate this possible lead if that quotient is sufficiently large. As we shall see, this number does not represent the likelihood ratio in favour of relatedness. 

In our opinion familial searching has the potential to lead to the same kind of debate as the classical database controversy. As we will point out in the text, such a classical database search can in fact be viewed as a special case in the model that we present.

Finally, let us mention that this paper can be viewed as a sequel to \cite{eiland}, in which we have discussed the effectiveness of database searches where one is looking for a perfect match. In that case, only two likelihood ratios arise as the result of the comparison of the target's profile to a database profile: zero (in case of exclusion) or $1/p$ (in case of a match, and $p$ is the profile frequency). In this article we generalize this to the search for a database member that has a certain family relation to the target (including the possibility of identity).

\section{Database likelihood ratios}
\label{secDB}
We consider a database of DNA profiles, and we suppose at most one of these profiles comes from a person that is related to the offender, whose DNA profile we call the target profile. The form of relatedness is fixed: we are looking for a specific relative (e.g., the offender's father, the offender's brother, or simply the offender himself), not for {\it any} relative. We also assume that the database members are unrelated to each other. In reality, the database may of course contain several related individuals. We expect this to have a negligible influence when these people are only related to each other and not to the target. In the case however where the database contains more than one relative of the target, the situation becomes more complex. Nonetheless it is clear that the target becomes easier to identify in that case. Exactly how much easier will depend on the number and type of these relatives and warrants a separate publication. In this article we restrict ourselves to the case where at most one relative may be present in the database, and we believe that the results are indicative of those to be obtained in the more general setting.

Let $\rv{P_1},\dots,\rv{P_N}$ represent the DNA profiles in database $\D$, by letting them be random variables that are all distributed as either $\rv{S}$ (``special") or $\rv{G}$ (``generic"), such that at most one of the $\rv{P_i}$ (the ``special" member) is distributed as $\rv{S}$. For example, if we are looking for brothers then ${\rv S}$ selects a DNA profile from a brother of the target, according to the probability distribution for DNA profiles of brothers of the target. By ${\rv R}=i$ we mean that ${\rv P_i}$ is distributed according to $\rv{S}$, for $1 \leq i \leq N$. By ${\rv R} \notin \D$ we mean that all ${\rv P_i}$ for $1 \leq i \leq N$ are distributed according to ${\rv G}$, i.e., the database does not contain a special item. Then ${\rv R} \in \D$ has the obvious meaning. We assume that we know the ``prior" probabilities \[ \pi_i=P(\rv{R}=i)\] for all $i$. Further, let $\pi_\D=\sum_{i=1}^N\pi_i$ and $\pi_0=1-\pi_D$. 

We also suppose that all $\rv{P_i}$ are conditionally independent given which of them (if any) is distributed as $\rv{S}$. This reflects that all population members are unrelated to each other and to the offender, except for the relative. 

In our DNA context, the $\rv{P_i}$ take values in the set of DNA profiles. The distribution of $\rv{S}$ will depend on the relation between the offender and the special member. If we are looking for the offender himself, then $\rv{S}$ will have a point distribution: it is equal to the DNA profile of the offender. If we are looking for the offender's father, then $\rv{S}$ will be equal to DNA profile $g$ with probability equal to the probability that the offender's father has DNA profile $g$. On the other hand, the random variable $\rv{G}$ always represents a DNA profile of a random population member, and therefore the probability that $\rv{G}=g$ is equal to the expected population frequency of DNA profile $g$, irrespective of the type of search we are performing (looking for a direct match, a parent or child, sibling, et cetera.)


The likelihood ratio $\rv{LR}_{\rv{S},\rv{G}}(\rv{P_i})$ expresses the support in favour of $\rv{R}=i$: it says how many times more likely it is that database member $i$ has the observed DNA profile if it is equal to the desired relative, than if it is unrelated to the offender. Such a likelihood ratio is often called PI (Paternity Index) when looking for a parent or child, and SI (Sibling Index) when looking for a sibling.

If every database member is subjected to such a likelihood ratio calculation, we obtain a vector of likelihood ratios, that we denote by
\[ \rv{LR_D}=(r_1,\dots,r_N)={\bf r}.\]
We also define 
\[ r_0=1,\]
which will be useful below.

We will use these likelihood ratios from the database to make probabilistic statements concerning the identity of $\rv{R}$. In the next chapter we will use those results to define search strategies for $\rv{R}$ in the database. One strategy relies on the following result, whose proof we defer to the Appendix.
\begin{prop}
\label{postpi}
For $i=1,\ldots, N$, we have
\begin{equation}
\label{nummereen}
P(\rv{R}=i \mid \rv{LR_D}=\r) = \frac{r_i \pi_i}{\sum_{k=0}^Nr_k\pi_k}
\end{equation}
and
\begin{equation}
\label{nummertwee}
P(\rv{R}=i \mid \rv{LR_D}=\r, \rv{R} \in \D) = \frac{r_i\pi_i}{\sum_{k=1}^Nr_k\pi_k}.
\end{equation}
\end{prop}

\begin{rem}
{\rm A version of the above formula (\ref{nummereen}) has, in the context of DNA mixtures, also been obtained in \cite{CHF} (see also \cite{CHF2}). Their formulation uses probabilities rather than likelihood ratios in numerator and denominator.}
\end{rem}

\begin{rem}
\label{opmerking}
{\rm Note that the conditional probabilities in (\ref{nummereen}) and (\ref{nummertwee}) do not depend on the distribution of $\rv{S}$ and $\rv{G}$.} 
\end{rem}
It follows from Proposition \ref{postpi} that, for any subset $\D' \subset \D$
\begin{equation}
\label{post1}
P(\rv{R}\in \D'\mid \rv{LR_D}=\r) = \frac{\sum_{i\in \D'}r_i\pi_i}{\sum_{k=0}^Nr_k\pi_k}.
\end{equation}
In particular, with $\D'=\D$ we obtain
\[ \frac{P(\rv{R} \in \D\mid \rv{LR_D}=\r )}{P(\rv{R} \notin \D\mid \rv{LR_D}=\r )}=\frac{\sum_{i=1}^Nr_i\pi_i}{\pi_0},\]
and the likelihood ratio in favour of $\rv{R}\in\rv{D}$ is given by
\begin{equation}
\label{LRD} \frac{P(\rv{LR_D}=\r \mid \rv{R}\in\D)}{P(\rv{LR_D}=\r \mid \rv{R}\notin \D)}=\frac{\sum_{i=1}^Nr_i\pi_i}{\pi_\D}.
\end{equation}
Note that the likelihood ratio depends on the prior probabilities. 

\begin{cor}
In odds form, we obtain
\[ \frac{P(\rv{R}=i \mid \rv{LR_D}=\r)}{P(\rv{R}\neq i \mid \rv{LR_D}=\r)}=\frac{r_i \pi_i}{\sum_{k=0, k \neq i}^Nr_k\pi_k},\]
and the likelihood ratio in favour of $\rv{R}=i$ is given by
\begin{equation} \label{persLR} \frac{P(\rv{LR_D}=\r \mid \rv{R}=i)}{P(\rv{LR_D}=\r \mid \rv{R}\neq i)} = \frac{r_i(1-\pi_i)}{\sum_{k=1, k \neq i}^Nr_k\pi_k+\pi_0}. 
\end{equation}
\end{cor}

In the case where the prior distribution of $\rv{R}$ on $\D$ is uniform, the above derived formulas simplify, and for convenience we include them here. With $P(\rv{R}=i)=\pi_\D/N$ for all $1\leq i \leq N$ we obtain, as special case of \eqref{LRD},
\begin{equation}
\label{uni}
\frac{P(\rv{LR_D}=\r \mid \rv{R}\in\D)}{P(\rv{LR_D}=\r \mid \rv{R}\notin \D)}=\frac{1}{N}\sum_{i=1}^N r_i,
\end{equation}
which is independent of the prior $\pi_{\D}$, contrary to the general case. In this uniform case, we see that the results $\rv{LR_D}=\r$ favour $\rv{R} \in \D$ if and only if the \emph{average} likelihood ratio on $\D$ is greater than one.

The posterior probability that $\rv{R}=i$ is now equal to
\[ P(\rv{R}=i \mid \rv{LR_D}=\r) = \frac{r_i}{\sum_{k=1}^Nr_k+N\frac{1-\pi_{\D}}{\pi_{\D}}},\]
with corresponding likelihood ratio
\begin{equation}
\label{LRPi} \frac{P(\rv{LR_D}=\r \mid \rv{R}=i)}{P(\rv{LR_D}=\r \mid \rv{R}\neq i)}=\frac{r_i}{\frac{\pi_{\D}}{N-\pi_{\D}}(\sum_{k=1, k \neq i}^Nr_k)+\frac{N(1-\pi_\D)}{N-\pi_{\D}}}.
\end{equation}
In the even more specific case that $\pi_D=1$, i.e., the database surely contains a relative and any of the members can be the relative with equal a priori probability, we simply get
\[ P(\rv{R}=i \mid \rv{LR_D}=\r)=\frac{r_i}{\sum_{k=1}^Nr_k}\]
and
\[
\frac{P(\rv{LR_D}=\r \mid \rv{R}=i)}{P(\rv{LR_D}=\r \mid \rv{R}\neq i)}=\frac{r_i}{\frac{1}{N-1}\sum_{k \neq i}r_k}.
\]

\begin{example}
{\rm We can view the process of searching for a match with a DNA profile as a special case. In that case $\rv{R}$ corresponds to the trace donor, and $\rv{S}$ can only take value $e_0$, the DNA profile in question. On the other hand $\rv{G}$ can take more values with probabilities given by the profile population frequencies. Let $p$ denote $P(\rv{G}=e_0)$, the random match probability of the profile $e_0$.
In this situation, $\rv{LR}(\rv{P_i})$ can take value $0$ or $1/p$. The variables $\rv{P_1},\dots,\rv{P_N}$ correspond to the members of a database in which we look for the profile $e_0$.  Suppose that the $i^{\rm th}$ database member is the only one that matches. Then $\rv{LR}(\rv{P_i})=1/p$ and all other $\rv{LR}(\rv{P_k})=0$ (for $1 \leq k \leq N, k \neq i$), so the likelihood ratio \eqref{LRPi} in favour of $\rv{R}=i$ becomes
\[ \frac{P(\rv{LR_D}=\r \mid \rv{R}=i)}{P(\rv{LR_D}=\r \mid \rv{R}\neq i)} =\frac{1}{p}\frac{N-\pi_{\D}}{N(1-\pi_{\D})}.\]
If $\pi_{\D}=N/n$ (the population fraction in the database), then this reduces to $(n-1)/(p(n-N))$, and
the likelihood ratio in favour of $\rv{R} \in \D$ is given by (cf.\ \eqref{LRD})
\[ \frac{P(\rv{R}=i)}{pP(\rv{R} \in \rv{D})}=\frac{1}{Np}.\]
These results are well known; see e.g.\ \cite{eiland} and the references therein.}
\end{example}

\section{Search strategies}\label{strat} We will now use these results to define strategies to choose a subset of $\D$ as small as possible, and which contains $\rv{R}$ with a given minimal probability $\alpha$.

\subsection{The conditional method} Let $\D^k$ be the subset of $\D$ that corresponds to the $k$ largest products $r_i\pi_i$ (with some arbitrary rule in case of ties). 
Furthermore, we let, for $ 0 \leq \alpha \leq 1$, $k_{\alpha}$ be the minimal $k$ for which
$$
\sum_{j \in \D^k} r_j\pi_j\geq \alpha (r_1\pi_1 + \cdots + r_N\pi_N),
$$
that is, $k_{\alpha}$ is the smallest $k$ for which the corresponding sum of the likelihood ratios weighted with the prior probabilities is at least a fraction $\alpha$ of the total weighted sum. Finally, we write $\D^{\alpha}$ for  $\D^{k_{\alpha}}$. Note that in order to determine whether or not $i \in \D^{\alpha}$, one needs the full vector $\r$.
Note also that $P(\rv{R} \in \D^\alpha)$ depends on the distribution of $\rv{S}$ and $\rv{G}$, but $P(\rv{R} \in \D^{\alpha} \mid
\rv{LR_D} =\r)$ does not (cf.\ Remark \ref{opmerking}). The distribution of the cardinality of $\D^{\alpha}$ also depends on the distribution of $\rv{S}$ and $\rv{G}$ (and on $N$).

We now make two observations about the probability that the index $\rv{R}$ is contained in $\D^{\alpha}$. First, 
in case $P(\rv{R} \in \D)=1$ it follows from (\ref{nummereen}) that for the unconditional probability $P(\rv{R} \in \D^{\alpha})$ we have
\begin{equation}
\label{groteralpha}
P(\rv{R} \in \D^{\alpha}) \geq \alpha.
\end{equation} 
Secondly, if we do not have $P(\rv{R} \in \D)=1$, then from (\ref{nummertwee}) we have (for $i=1,\ldots, N$) that
\begin{eqnarray*}
P(\rv{R}=i \mid \rv{R} \in \D ) &=& \sum_{\r} P( \rv{LR_D}=\r\mid\rv{R} \in \D)P(\rv{R}=i\mid\rv{R} \in \D, \rv{LR_D}=\r )\\
&=&\sum_{\r}P(\rv{LR_D}=r\mid \rv{R} \in \D ) \frac{r_i\pi_i}{\sum_{k=1}^N r_k \pi_k}.
\end{eqnarray*}
Hence
\begin{eqnarray*}
P(\rv{R} \in \D^\alpha \mid \rv{R} \in \D)&=&\sum_{\r}P(\rv{LR_D}=\r\mid\rv{R} \in \D) P(\rv{R} \in \D^\alpha \mid \rv{R} \in \D, \rv{LR_D}=\r)\\
&=& \sum_{\r} P(\rv{LR_D}=\r\mid\rv{R} \in \D ) \sum_{i \in \D^\alpha} \frac{r_i\pi_i}{\sum_{k=1}^N r_k \pi_k}\\
&\geq & \alpha \sum_{\r} P(\rv{LR_D}=\r\mid\rv{R} \in \D ) =\alpha,
\end{eqnarray*}
where the inequality follows from the definition of $\D^{\alpha}$.
The quantity
$$
P(\rv{R} \in \D^\alpha | \rv{R} \in \D)
$$
is called the {\em efficiency} of $\D^\alpha$, the ability to select $\rv{R}$ given that $\rv{R}$ is in the database.  We just showed that the efficiency of $\D^\alpha$ is at least $\alpha$.

\subsection{The target-centered method} Recall that we have described $\rv{S}$ as the random variable that selects DNA profiles of relatives, and $\rv{G}$ as the random variable that selects DNA profiles of unrelated individuals. These being random variables, the likelihood ratios in favour of relatedness for related individuals (drawn from $\rv{S}$) and for unrelated individuals (drawn from $\rv{G}$) also become random variables. By $P(\rv{LR}(\rv{S})\geq t)$ we mean the probability that the likelihood ratio for a related individual is at least $t$. For example, when looking for a sibling, this would correspond to the probability that the sibling index with a random sibling of the offender is at least $t$. Similarly, it also makes sense to write $\rv{LR}(\rv{P_i})$: this is the likelihood ratio obtained from population member $i$.

For $0 \leq \alpha \leq 1$, let $t_\alpha \geq 0$ be the largest $t$ for which
\begin{equation}
\label{t}
P(\rv{LR}(\rv{S}) \geq t)\geq \alpha.
\end{equation}

We use these thresholds $t_\alpha$ to define
\begin{equation}
\label{dalpha}
\D_\alpha=\{ i \in \D \mid \rv{LR}(\rv{P_i}) \geq t_\alpha\}.
\end{equation}
In order to decide whether or not $i \in \D_{\alpha}$, one only needs to know $r_i$, and not the full vector $\r$ as in the case of
$\D^{\alpha}$.
It follows that 
\[ P(\rv{R} \in \D_\alpha \mid \rv{R} \in \rv{D})=P(\rv{LR}(\rv{S})\geq t_\alpha) \geq \alpha,\]
so also the efficiency of $\D_\alpha$ is at least $\alpha$. In fact, for every $0 \leq \alpha \leq 1$, 
\begin{equation}
\label{post2}
 P(\rv{R} \in \D_\alpha)\geq \alpha P(\rv{R}\in \D).
\end{equation}

Thus, we simply choose a threshold that is met by the real relative with probability $\alpha$, and admit anyone into $\D_\alpha$ when the likelihood ratio is sufficiently big. 

\begin{rem} {\rm Conversely, one could also select all database members whose likelihood ratio is unusually big if they would be unrelated.
This criterion has been proposed in \cite{sjerpskloos}, in the context of deciding whether or not to further investigate the possibility that the suspect's relative matches a crime stain. In our terminology, they propose a threshold $s_\beta$ such that
\[ \beta=P(\rv{LR(G)} \geq s_\beta).\]
If $N$ is large then (since all but at most one of the database members are distributed as $\rv{G}$) one expects a fraction $\beta$ of the database to be selected into $\{ i \in \D \mid \rv{LR(P_i)} \geq s_\beta\}$. The relation with $\D_\alpha$ is as follows. We have 

\begin{eqnarray*} \alpha &\leq & P(\rv{LR(S)}\geq t_\alpha) \\ &=& \sum_{x \geq t_\alpha} P(\rv{LR(S)}=x) \\ &=& \sum_{x \geq t_\alpha} x P(\rv{LR(G)}=x) \\  &=& \sum_{x \geq 0} xP(\rv{LR(G)}=x \mid \rv{LR(G)} \geq t_\alpha)P(\rv{LR(G)}\geq t_\alpha) \\ &=& E(\rv{LR(G)} \mid \rv{LR(G)} \geq t_\alpha )P(\rv{LR(G)} \geq t_\alpha).\end{eqnarray*}
Now, let $\beta$ be such that $t_\alpha=s_\beta$, then
\[ \alpha \leq \beta\cdot E(\rv{LR(G)} \mid \rv{LR(G)} \geq s_\beta).\]
Clearly, $\alpha$ cannot be expressed in $\beta$ alone but depends also on the target that we are dealing with. It follows that when selecting a database subset according to the threshold $s_\beta$, the probability that $\rv{R}$ is selected depends on the specific aspects of the case, whereas for $\D_\alpha$ (and $\D^\alpha$) it has a uniform lower bound $\alpha$.}
\end{rem}

\begin{rem} {\rm Notice that this approach essentially compares false-negative probabilities ($1-\alpha$) with false-positive probabilities ($P(\rv{LR(G)}\geq t_\alpha)$). This is also done in \cite{ge}, but without focussing on a specific target profile. The results are therefore informative about the average performance with likelihood ratio thresholds, not about results in a particular case. Moreover, various studies (cf. \cite{ge},\cite{CowenThomson}) employ thresholds that are a combination of a threshold on the number of shared alleles (denoted IBS, identity by state) and a likelihood ratio threshold $t$. Suppose we define $\D^{n,t}$ as those database profiles that share at least $n$ alleles with the target profile and for which the kinship index with the target profile is at least $t$. For a given target, $\D^{n,t}$ will contain a true relative with a certain probability, and will contain unrelated individuals with another probability. However, the Neyman-Pearson lemma implies that of all tests with the same false-negative rate, the likelihood ratio test has the smallest false-positive rate. In other words, the pair $(n,t)$ will correspond to a false-negative rate $1-\alpha$, and then $\D^\alpha$ which has this same false-negative rate, will have a better (more precisely, not a worse) false-positive rate than $\D^{n,t}$. Therefore, even though it can be computationally attractive to work with a threshold $(n,t)$, conceptually these thresholds are outperformed by putting a threshold on the likelihood ratio alone.}
\end{rem}

\subsection{Comparison and interpretation}\label{comp}
We have defined two subsets $\D^\alpha$ and $\D_\alpha$, both with efficiency at least $\alpha$.
Nevertheless, there are important differences between these approaches that we wish to discuss here. 

First of all, $\D^\alpha$ makes use of the prior probabilities $\pi_i=P(\rv{R}=i)$, while $\D_\alpha$ does not. For example, in case of familial searching, geographical information or age could play a role in the definition of prior probabilities $P(\rv{R}=i)$. Thus, $\D^\alpha$ uses more information than $\D_\alpha$, which seems to give $\D^\alpha$ an advantage over $\D_\alpha$. 

There is, however, a reason why the use of $\D_\alpha$ could be more appropriate in concrete cases. This reason has to do with the interpretation of the probabilities involved, and we explain this next. We can see $\D^\alpha$ as a random subset of $\D$ which contains all database members that have yielded likelihood ratios greater than or equal to a {\em random} threshold. The distribution of this threshold depends on the distributions of both $\rv{S}$ and  $\rv{G}$ (and on $N$, the size of the database). Therefore, a frequentist interpretation requires re-sampling of the database. Indeed, we have defined a subset in such a way that, if we would construct it for many realizations of one copy of $\rv{S}$ among $N-1$ copies of $\rv{G}$, a fraction $P(\rv{R}\in \D^\alpha \mid \rv{R} \in \rv{D})$ of the time we would have included the copy of $\rv{S}$.

The interpretation of the probability $P(\rv{R} \in \D_\alpha)$, on the other hand, is easier. Indeed, $\D_\alpha$ is a random subset of $\D$ as well, containing all database members that have yielded likelihood ratios above some threshold, but this time the threshold depends on the distribution of $\rv{S}$ only (through $\rv{LR(S)}$). This allows us to make another frequentist interpretation: we choose a realisation of the database (according to $\rv{G}$), and then, keeping the database fixed, repeatedly add one copy of a realisation of $\rv{S}$. We can think of $P(\rv{R} \in \D_\alpha| \rv{R} \in \D)$ as the relative frequency of times we would find the special member in $\D_\alpha$. This interpretation corresponds well with what one would intuitively understand by the probability of finding the relative since in the forensic practice, the database is (more or less) fixed. From this point of view it is more appropriate to use $\D_\alpha$ rather than $\D^\alpha$ and, importantly, it is also easier to explain to legal representatives what the probabilistic statement really means. 

The frequentist considerations above apply to the general framework we have discussed in this paper. In the special case of familial searching however, the drawback of using $\D^\alpha$ may not be that serious, for the following reason. We explained that for a full frequentist interpretation of $\D^\alpha$, one would need to resample the database many times, and that this does not correspond well to legal practice. However, what matters is not so much that we can interpret the full profiles in the database as being resampled, but that we can interpret the 
{\em observed likelihood ratios} as being resampled. Suppose that we treat various different familial searching cases (i.e., try to find relatives from various targets) with the same database. Then, when we compute likelihood ratios between the database and a new target, these likelihood ratios depend on the newly sampled target profile, and it is to be expected that for an independent sequence of target profiles, the observed likelihood ratios corresponding to the fixed profiles in the database are more or less independent. To test this, in the next section we investigate using computer simulation to what extent the frequentist interpretation that we have for $\D_\alpha$ is valid for $\D^\alpha$ as well. That is, we draw many targets independently according to $\rv{G}$ (i.e., at random using population allele frequencies), and add their simulated relatives to the same database. We see how many of these relatives are found in $\D^\alpha$, on average over all targets. This we will compare to adding many relatives of the {\em same} target to resampled databases.

Finally, we mention the fact that when the database is large and uniform priors are used, the sets $\D^\alpha$ and $\D_\alpha$ will be very similar. This is due to the fact that the law of large numbers implies that the sum of the likelihoods above $t_\alpha$ divided by $N$ will be close to $\alpha$. Hence the random threshold associated with $\D^\alpha$ (discussed above) will with very high probability be very close to $t_\alpha$. This argument can be made precise in the form of a limit statement in probability or almost surely.  

\subsection{Heterogeneous databases}
A forensic DNA database may consist of profiles that have different sets of loci typed. As a result, not all profiles stored in the database are equally informative. Mathematically this means that instead of one pair $(\rv{S},\rv{G})$, we have several such pairs. We now sketch how the search strategies deal with such a situation. For ease of exposition, we consider only the situation where there are two such pairs, the generalization to a larger number being obvious. In the DNA context, this corresponds to a database that contains DNA profiles for two different sets of loci. We write the database $\D$ as a disjoint union $\D=\D_1\cup\D_2$, where $\D_i$ corresponds to the hypothesis random variables $\rv{S}_i$ and $\rv{G}_i$. 

\subsubsection{Conditional method} 
Going through the proof of Proposition \ref{postpi}, one checks that this expression also holds in this heterogeneous situation, with the understanding that if $\rv{P_i} \in \D_j$ then $\rv{P_i}$ is distributed either as $\rv{S_j}$ (if $\rv{R}=\rv{P_i}$) or as $\rv{G_j}$ (if $\rv{R} \neq \rv{P_i}$). The search strategy therefore need not be modified but its efficiency may change. For example, as the amount of genetic information in $\D_1$ increases, $\rv{S_1}$ approaches a point distribution at infinity. In the limit, if $\rv{R} \in \D_1$ then $\rv{R} \in \D^\alpha$ for all $\alpha>0$ and hence the efficiency of $\D^\alpha$ is at least $P(\rv{R} \in \D_1 \mid \rv{R} \in \D)$. Thus, the heterogeneity of $\D$ has an effect on the efficiency and cardinality of $\D^\alpha$.

Another possibility is to perform the conditional method on each database part separately with its own efficiency parameter $\alpha_i$; we may denote the resulting subset of $\D$ by $\D^{\alpha_1,\alpha_2}=\D_1^{\alpha_1} \cup \D_2^{\alpha_2}$. In that case,
\[ P(\rv{R} \in \D^{\alpha_1,\alpha_2} \mid \rv{R} \in \D) \geq \alpha_1 P(\rv{R} \in \D_1 \mid \rv{R} \in \D)+\alpha_2 P(\rv{R} \in \D_2 \mid \rv{R} \in \D).\]
In particular, choosing $\alpha_1=\alpha_2$ leads to an efficiency of (at least) $\alpha$ for $\D^{\alpha,\alpha}$. Other choices of $\alpha_1,\alpha_2$ may lead to the same efficiency. 
If, for example, $\D_1$ contains more genetic information than $\D_2$, then one may take advantage of this fact by letting $\alpha_1> \alpha_2$. This means that the efficiency is greater in $\D_1$ than it is in $\D_2$, while the overall efficiency of $\D^{\alpha_1,\alpha_2}$ is equal to $\alpha=(\alpha_1-\alpha_2)P(\rv{R} \in \D_1 \mid \rv{R} \in \D)+\alpha_2$. For these choices $\D^{\alpha,\alpha}$ and $\D^{\alpha_1,\alpha_2}$ have the same efficiency, but the expected cardinality of $\D^{\alpha_1,\alpha_2}$ can be smaller than that of $\D^{\alpha,\alpha}$.

\subsubsection{Target-centered method} In this case as well, we can make use of the database's heterogeneity to define various search strategies with the same efficiency. Let, as in \eqref{t}, $t_{i,\alpha} \geq 0$ be the largest $t$ for which $P(\rv{LR}(\rv{S_i}) \geq t)\geq \alpha$ (where $\rv{LR}=\rv{LR}_{\rv{S_i},\rv{G_i}})$, i.e., we define the appropriate thresholds on the likelihood ratio for every type of database entry. Then we let
\[ \D_{\alpha_1,\alpha_2}=\{i \in \D_1 \mid \rv{LR}(\rv{P_i}) \geq t_{1,\alpha}\} \cup \{i \in \D_2 \mid \rv{LR}(\rv{P_i}) \geq t_{2,\alpha}\},\]
containing a database member if, considering the amount of information stored for this database member, the obtained likelihood ratio is sufficiently big.
Then the efficiency of this strategy is simply, as it was for $\D^{\alpha_1,\alpha_2}$,
\[ P(\rv{R} \in \D_{\alpha_1,\alpha_2} \mid \rv{R} \in \D) \geq \alpha_1P(\rv{R} \in \D_1 \mid \rv{R} \in \D)+\alpha_2P(\rv{R} \in \D_2 \mid \rv{R} \in \D).\]
 Using the same limit as for the conditional method, as $\D_1$ contains more information, an efficiency of $\alpha_1=1$ can be obtained in the limit using $t_{1,1}=\infty$, in which case the efficiency of $\D_{\alpha_1,\alpha_2}$ will be at least $P(\rv{R} \in \D_1 \mid \rv{R} \in \D)$, the same as for the conditional method.

\section{Familial Searching in the Dutch National {\textsc dna} Database}
\label{secFS}
\subsection{Methods and notation}\label{methods}
We have carried out simulation experiments using the Dutch DNA database, where we have carried out familial searches with artificial targets, looking for parent-child, sibling, and half-sibling relations. We will restrict ourselves mostly to the results of the sibling and half-sibling searches here, parents and children being relatively easy to find owing to the fact that they always share an allele on each locus (barring mutations, but these are rare) whereas siblings (and of course half siblings too) need not share any alleles with each other.

All our simulations were programmed in-house with Mathematica software. We let $\D_{NL}$ be the Dutch National DNA Database (as per mid 2010, all duplicate profiles removed and only considering the $N =99,979$ profiles for which all ten SGMPlus loci were typed). Allelic ladders and allele frequencies were taken from $\D_{NL}$.

According to these allele frequencies, target profiles $C_1, \dots, C_{100}$ were sampled (pseudo)randomly to serve as the targets whose relatives we want to find using familial searching. For each of these target profiles we sampled 50,000 children and 50,000 siblings. Then we computed the likelihood ratios in favour of paternity (the {\em Paternity Index} $PI$) between the $C_i$ and their children, and those in favour of siblingship (the {\em Sibling Index} $SI$) between the $C_i$ and their siblings.  These allow us to estimate the thresholds $t_\alpha$ (cf.\ \eqref{t}) for the paternity and sibling cases. 

We will sometimes write $KI$, for {\em Kinship Index}, when we mean that the discussion holds for any type of relative, in particular $KI$ can stand for $PI$ or $SI$.

The DNA profiles in $\D_{NL}$ are labeled $d_1,\dots,d_N$; they can be viewed as a sample of independent copies of $\rv{G}$ that is fixed throughout. By $KI(C_i,d_j)$ we mean the kinship index between the target profile $C_i$ and database profile $d_j$. Thus, $KI(C_i,d_j)$ can, for each target separately, be interpreted as a realization $r_j$ of the random variables $\rv{LR(G)}$ in the preceding sections. 

We have also computed the random match probability (RMP) of each target profile. On a locus with alleles $(a,b)$, the RMP is equal to $p_ap_b(2-\delta_{a,b})$ where $p_i$ is the allele frequency of allele $i$ and $\delta_{a,b}=0$ if $a\neq b$ and $\delta_{a,a}=1$. The RMP of a DNA profile is then the product over all involved loci, since we assume all loci to be independent. This is reasonable, since all loci are on different chromosomes.

\subsection{Total likelihood ratio with the database} 
For all 100 targets $C_i$, we computed the sums $|KI(C_i,\D_{NL})|=\sum_{k=1}^NKI(C_i,d_k)$. The mean $|PI(C_i,\D_{NL})|$ was 102,200 (with sample standard deviation 94,500), the mean $|SI(C_i,\D_{NL})|$ was 93,500 (with sample standard deviation 42,200). 
These results seem consistent with what we expect, cf. Proposition \ref{e1}. Indeed, since all target profiles were randomly generated, they do not have a true relative in the database, and hence $E(|KI(C_i,\D_{NL})|)=N$ according to Proposition \ref{e1}.

\subsection{Conditional method}
First, we have compared the results of the conditional method when each familial search is performed in the same database (i.e., the Dutch National DNA database), with results obtained when each familial search is carried out in a new database. 
\subsubsection{The conditional method in the Dutch National DNA Database}
\label{effnatdb}
We have investigated (by simulation) what the probability is that a relative of a fixed target is found in $\tilde{\D}_{NL}^\alpha$, where $\tilde{\D}_{NL}$ is the extension of $\D_{NL}$ with a relative of the considered target. We take a uniform prior 
distribution of $\rv{R}$ on $\tilde{\D}_{NL}$. 

To do so, we have simulated relatives $R_{i,j}$ ($i=1,\dots,100; j=1,\dots,500$) of each type (children and siblings), where $R_{i,j}$ is a relative of target profile $C_i$. 
For each relative $R_{i,j}$, define its {\em rank} to be equal to $k$ if and only if there are exactly $k-1$ database members that have a greater kinship index with $C_i$ than $R_{i,j}$. 
We also define  
\[ t_{i,j}=\frac{\sum_{x: KI(C_i,d_x)>KI(C_i,R_{i,j})}KI(C_i,d_x)}{KI(C_i,d_1)+\dots+KI(C_i,d_N)+KI(C_i,R_{i,j})}.\]
Assuming a uniform prior of $\rv{R}$ on $\tilde{\D}_{NL}$, $t_{i,j}$ is the greatest $t\geq 0$ such that $R_{i,j} \notin \tilde{D}^{t}_{NL}$. Thus, $R_{i,j} \in \tilde{\D}^\alpha_{NL}$ if and only if $\alpha > t_{i,j}$. 

For each $C_i$ and for $\alpha \in \{0.01,\dots,0.99,1\}$, we have compared $\alpha$ to the fraction $\beta_{i,\alpha}$ of $t_{i,j}$ that are smaller than $\alpha$; this fraction $\beta_{i,\alpha}$ is the observed probability for relatives of $C_i$ to be in $\tilde{\D}^\alpha_{NL}$.
Finally, we have also computed $\beta_\alpha$ as the average over all $\beta_{i,\alpha}$. Thus, $\beta_\alpha$ estimates the probability that if one adds a relative $\rv{R}$ of a random target profile to {\it this} database $\D_{NL}$, that $\rv{R}$ is in $\tilde{\D}^\alpha_{NL}$. 

The probability that the relative of a target $C$ is in $\tilde{\D}^\alpha_{NL}$ is called the {\em probability of detection} (POD) for $C$ in the Dutch National Database $\D_{NL}$. The number $\beta_\alpha$ therefore gives an estimate of the average (over all targets) probabilities of detection POD. Note that a POD is only defined in connection to a fixed database, in this case $\D_{NL}$. 

For siblings, the result of our simulations is displayed in Figure \ref{SI-alpha}.
\begin{figure}[h!]
\caption{The average POD as a function 
of $\alpha$, average over 100 target profiles, Sibling Index.}
\label{SI-alpha}
\includegraphics[width=10cm]{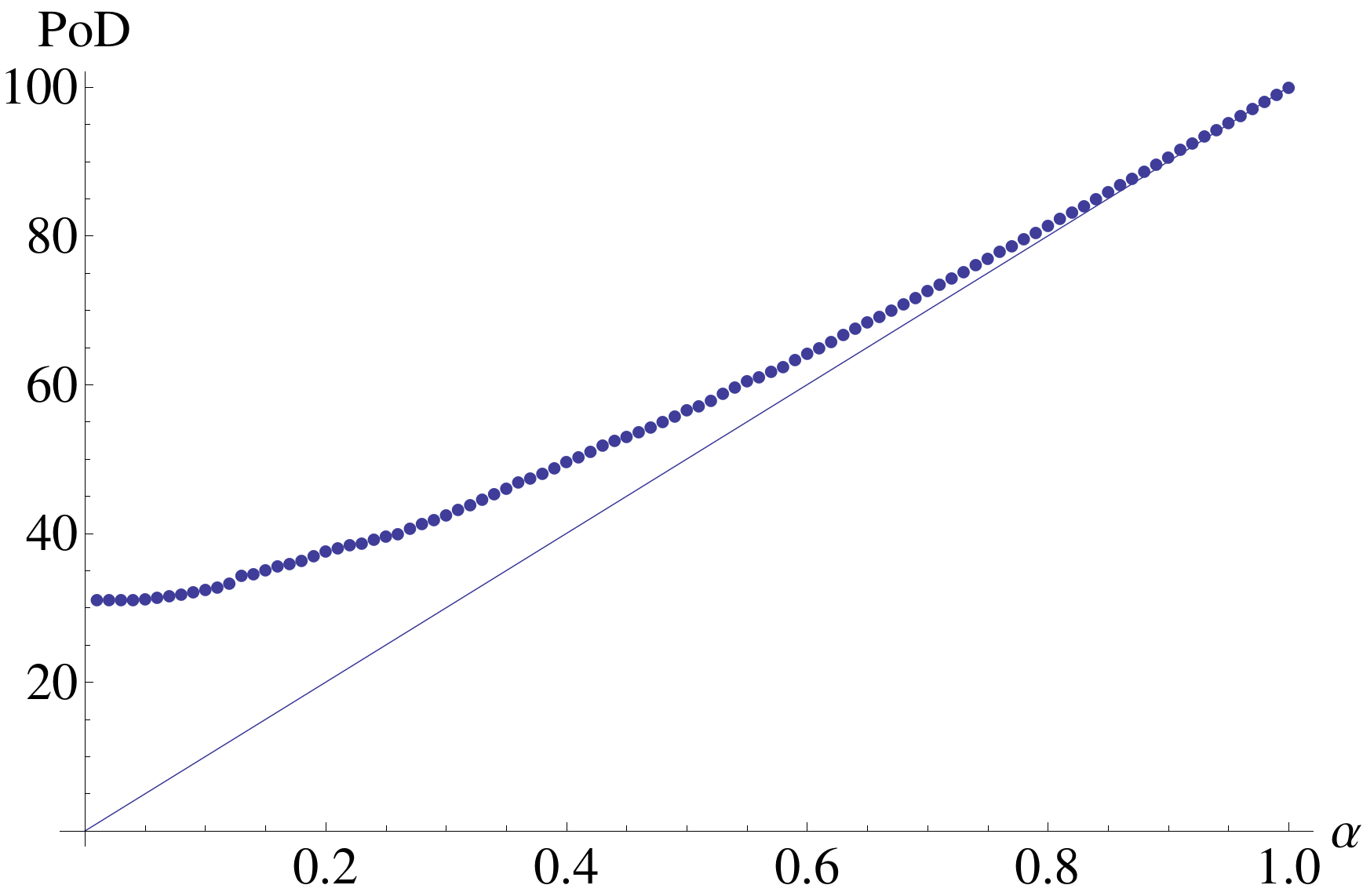}
\end{figure}

For all $\alpha$ the average POD of $\tilde{\D}^\alpha_{NL}$ is at least $\alpha$. For small $\alpha$, it exceeds $\alpha$ substantially and as $\alpha$ increases, the average POD of $\tilde{\D}^\alpha_{NL}$ approaches $\alpha$. This is a consequence of the definition of $\tilde{\D}^\alpha_{NL}$ as being a subset that contains the $k$ greatest $PI$ for some $k$. As $\alpha$ increases and $\tilde{\D}^\alpha_{NL}$ becomes larger, we add individuals with smaller $SI$, and we expect $\beta_\alpha$ to become closer to $\alpha$. However, the variation between target profiles was substantial. We highlight three very different results in Figure \ref{PI-alpha-3prof}.
\begin{figure}[h!]
\caption{The POD as a function of $\alpha$, for three target profiles.}
\label{PI-alpha-3prof}
\includegraphics[width=4cm]{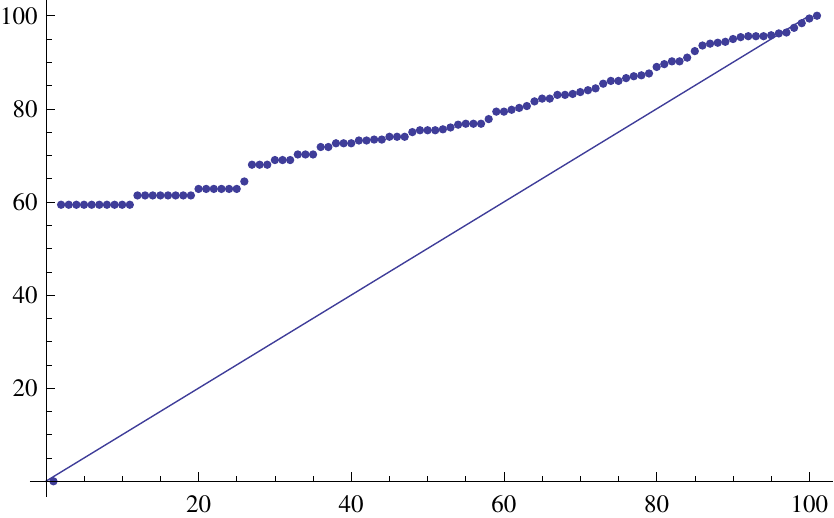}
\includegraphics[width=4cm]{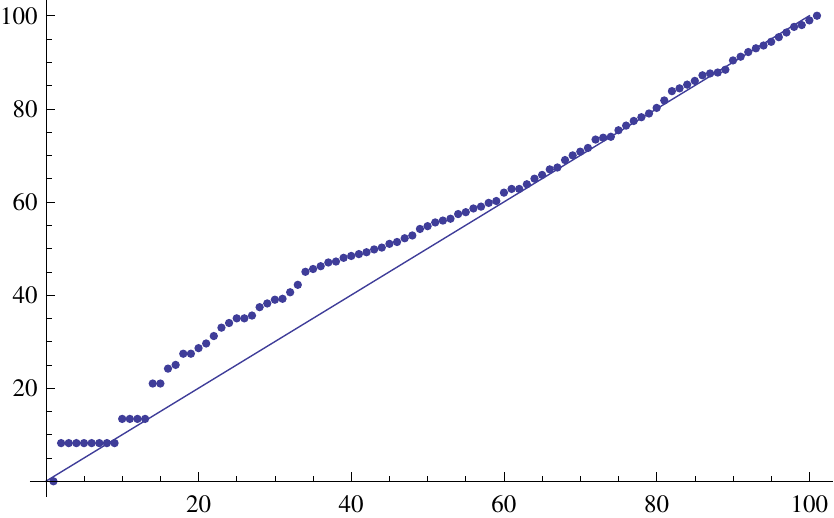}
\includegraphics[width=4cm]{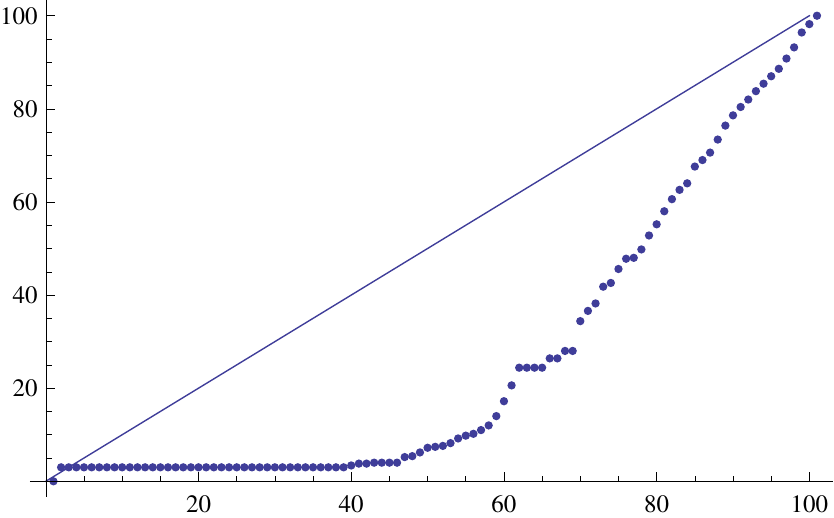}
\end{figure}
This difference is created by the presence (or absence) of database members that have a large sibling index with the target, thus obscuring (or revealing more easily) the real sibling, especially for low probabilities of detection.

\subsubsection{The conditional method in resampled databases}\label{condresamp}

In this section we compare the above estimates of the probabilities of detection with estimates of the efficiency of the conditional method in resampled databases of roughly the same size as $\D_{NL}$. To do so we have, for each target profile $C_i$ as above, simulated 100 relatives $R'_{i,j}$ and databases $\D_{i,j}$ with $N=100,000$. As in the previous section, we have determined $t_{i,j}$ as the largest $\alpha$ such that 
$R'_{i,j} \notin \tilde{\D}^\alpha_{i,j}$, and used these numbers to determine $\beta'_{i,\alpha}$ and $\beta'_\alpha$ whose definitions are analogous to their earlier counterparts. 

The observed overall efficiency $\beta'_\alpha$ is extremely close to the observed probabilities of detection $\beta_\alpha$ displayed in Figure \ref{SI-alpha}. In fact, the difference between the $\beta'_\alpha$ of this section (the average efficiency) and the $\beta_\alpha$ in Figure \ref{SI-alpha} (average probability of detection) is on average over $\alpha\in\{0.01,\dots,0.99,1\}$  equal to $-0.0022$ and never greater (in absolute value) than 0.0084.

We have also, for each $R_{i,j}$, computed its rank $k_{i,j}$ (defined in Section \ref{effnatdb}). A summary of the results is presented in Figure \ref{AlphaVersusRankAverage}, where we plot for each $\alpha \in \{0, 0.01, \dots, 0.99\}$, the average rank of relatives for which $t_{i,j}$ is nearest to $\alpha$. This gives the average rank of a relative that would be found in $\D^\alpha$, but not in $\D^\beta$ for $\beta < \alpha$.

\begin{figure}[h]
\caption{Observed pairs $(t_{i,j},{\rm Log}_{10}(\mbox{Mean }k_{i,j}))$}
\label{AlphaVersusRankAverage}
\includegraphics[width=10cm]{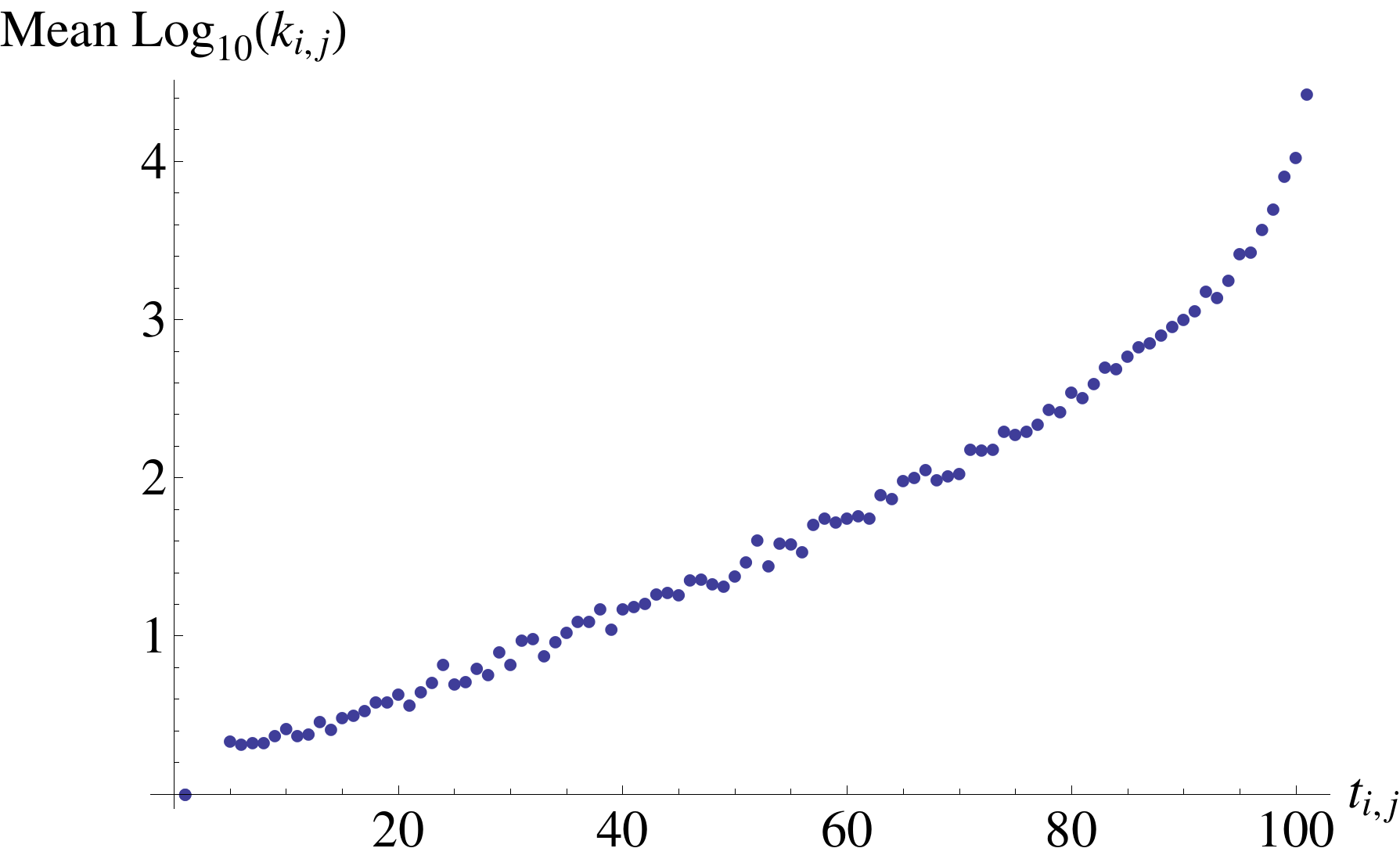}
\end{figure}

\subsubsection{Comparison of simulations}
The simulations indicate that, even though there is a conceptual difference between probability of detection and efficiency, the average probability of detection in the Dutch National DNA Database coincides with the average efficiency of $\D^\alpha$ for databases of the same size. 
In other words, the average probability of detection in $\D_{NL}$, averaged over all targets, is the same as the average efficiency of $\D^\alpha$ (averaged over all targets) for a database $\D$ of the same size as $\D_{NL}$. Therefore, even if we (as we do in practice) do not resample the database but look for relatives in the same database for all targets, then the probability of finding a relative in $\D^\alpha$ when database and relative are resampled is the same as the long term success rate of finding the relative of varying targets in $\D^\alpha$ while the database is kept constant.
On the other hand, for $\D_\alpha$, the interpretation of the efficiency $P(\rv{R} \in \D_\alpha \mid \rv{R} \in \D)$ does not require resampling of the database, hence also holds in a fixed one. This makes the $\D_\alpha$ method easier to interpret.


\subsection{Target-centered method} Recall that we have, for each of the target profiles $C_i$, simulated 50,000 siblings, in order to estimate the profile-dependent thresholds $t_\alpha$ needed for $\D_\alpha$. Some resulting sizes of $\D_{NL,\alpha}$, for $\alpha=0.70, 0.80, 0.90$ are plotted in Figure \ref{SImatches} where each dot represents a target profile. The horizontal axis contains $-{\rm Log}_{10}(RMP)$ with $RMP$ the random match probability (cf. end of \ref{methods}). The mean sizes of $\D_\alpha$ are 85, 258, 1038 respectively, and we notice a tendency for $\D_\alpha$ to be smaller for profiles with a smaller random match probability (i.e. for which $-{\rm Log}_{10}(RMP)$ is larger). This is to be expected: for such profiles $t_\alpha$ will be greater and it will be more unlikely for an unrelated person to have a sibling index with the target exceeding that threshold. The observed mean sizes of $\D_{NL,\alpha}$ are similar to what has been obtained in resampled databases with the conditional method (cf. Figure \ref{AlphaVersusRankAverage}).

\begin{figure}[h]
\caption{Size of $\D_{NL,\alpha}$ for $\alpha$=0.70, 0.80 and 0.90, Sibling Index}
\label{SImatches}
\includegraphics[width=4cm]{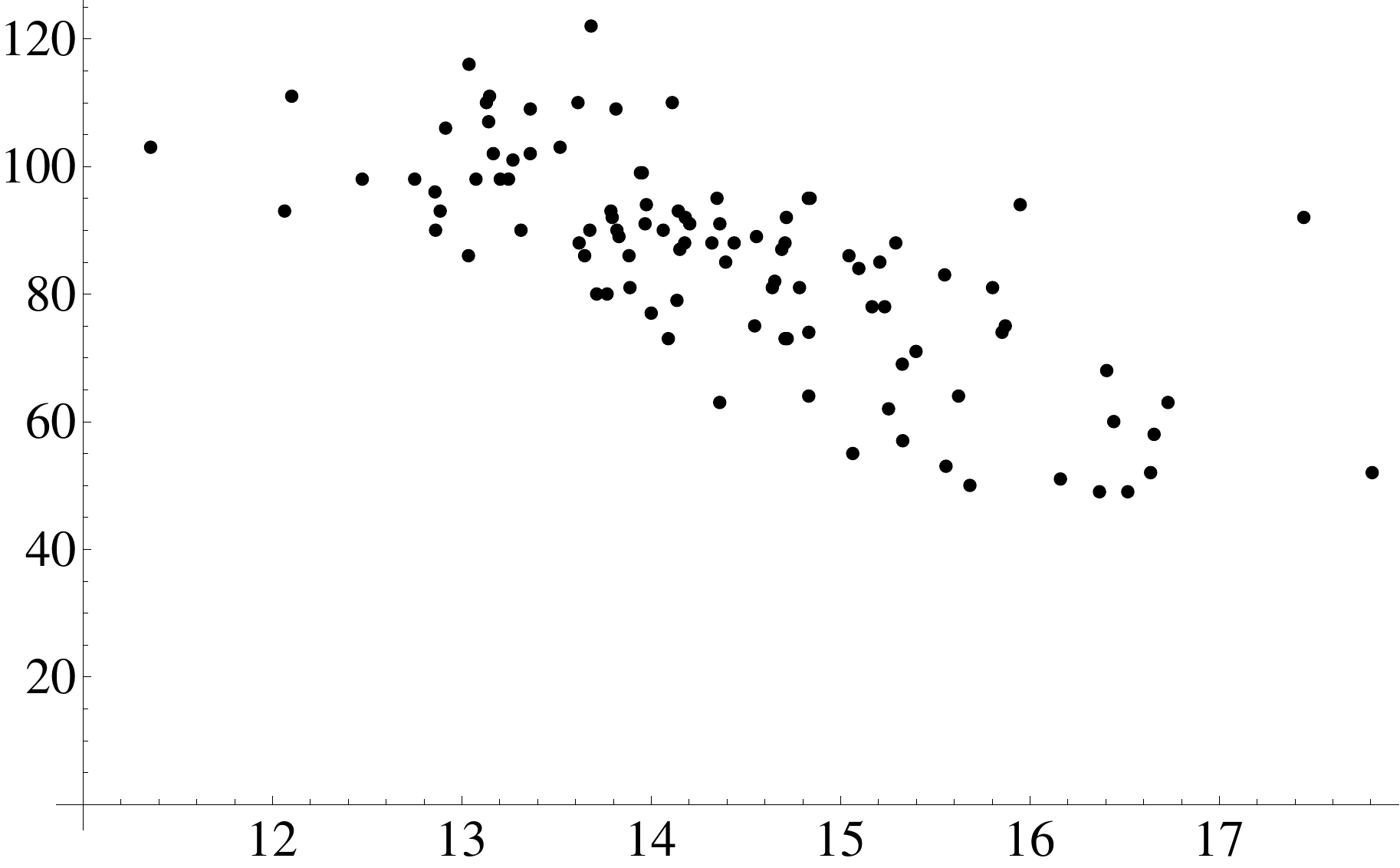}
\includegraphics[width=4cm]{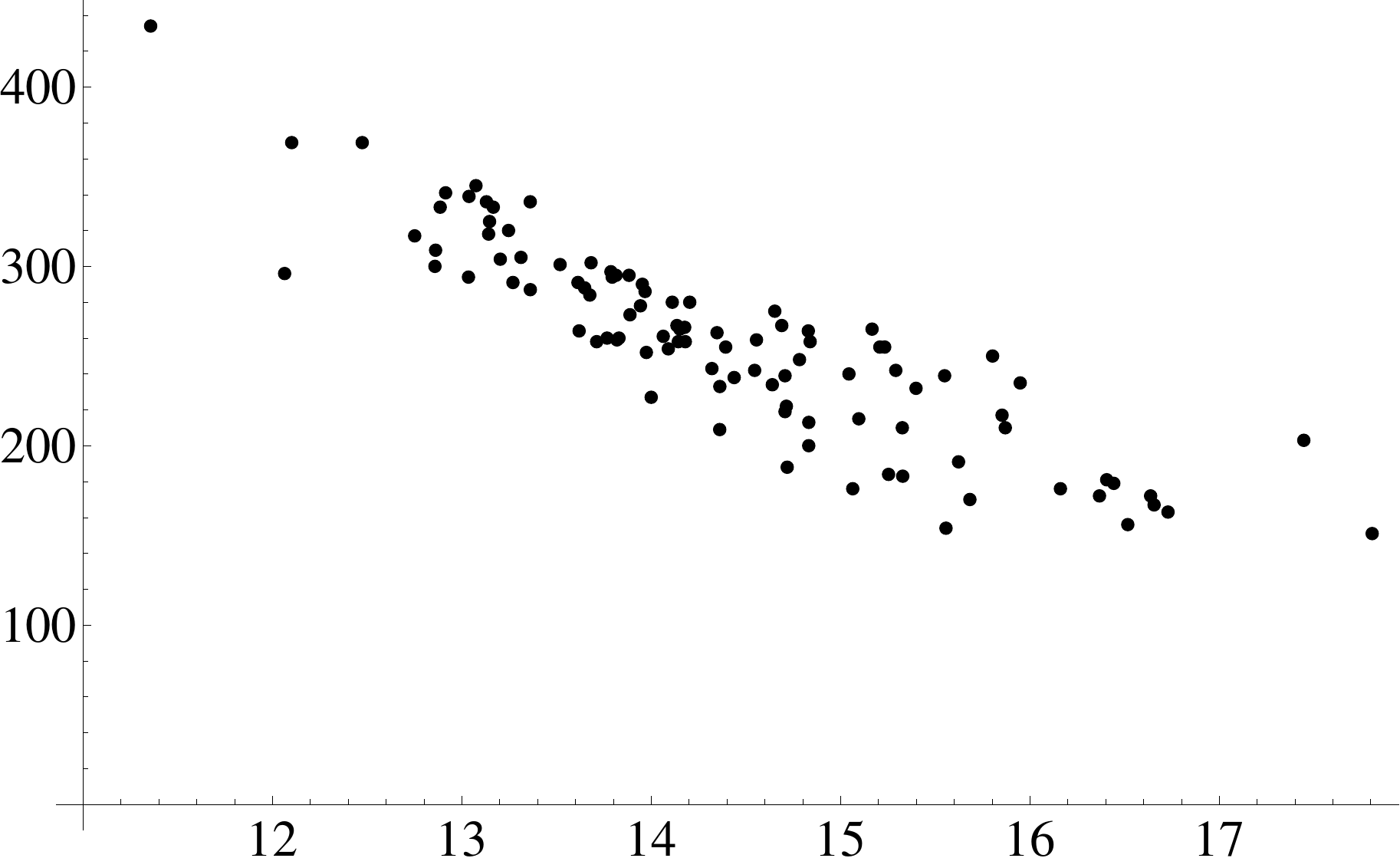}
\includegraphics[width=4cm]{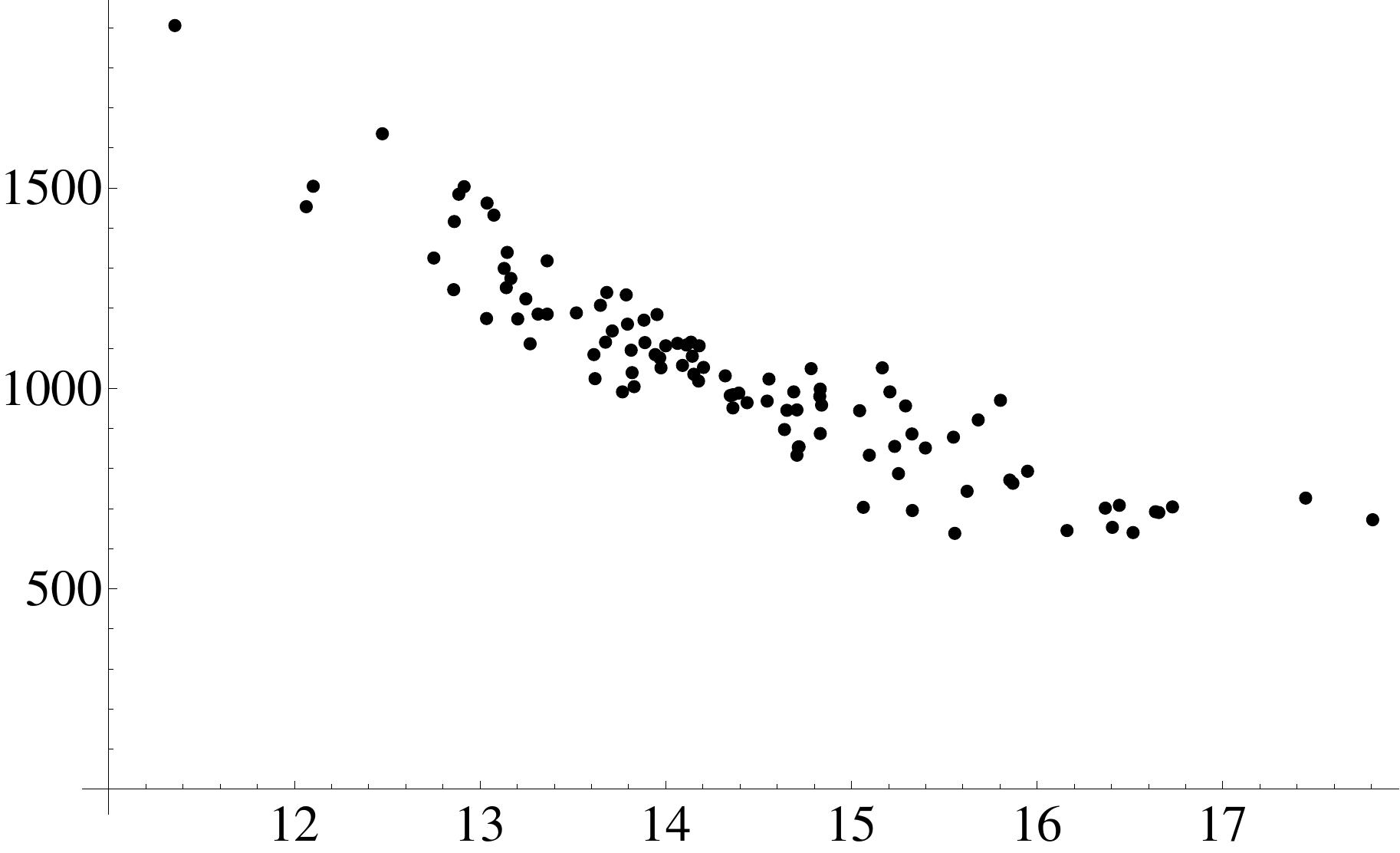}
\end{figure}


\subsection{Rank in Dutch database}
\label{ranking}
Another possibility for a forensic lab is to fix in advance the number of individuals that are additionally tested. Especially when doing so, it is interesting to know how high one expects the relative to rank when the database is sorted according to decreasing kinship index. For each of the target profiles $C_i$, we have generated 1000 siblings and investigated what their rank would be if it would have been added to the Dutch National DNA database. The result, averaged over all targets, is displayed in Figure \ref{AvRank}. This graph should be compared to Figure \ref{AlphaVersusRankAverage}.

\begin{figure}[h]
\caption{Probability for a sibling to obtain SI-rank at most $n$, in Dutch database}
\label{AvRank}
\includegraphics[width=12cm]{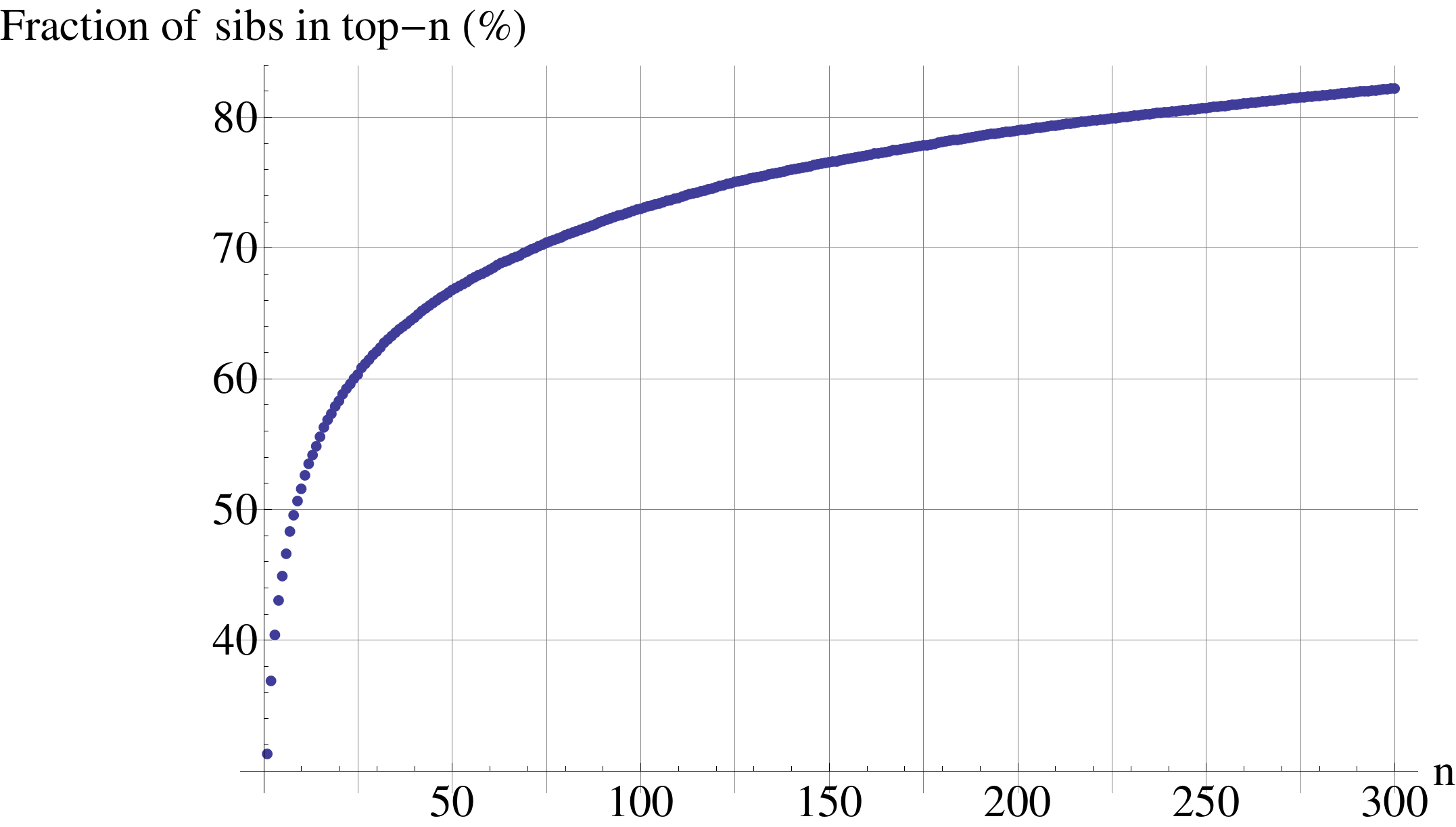}
\end{figure}

In particular, the true sibling was ranked first in 31\% of the cases, in the top 10 in about 52\% of the cases, in the top 100 in about 73\% of the cases, and in the top 200 in about 79\% of the cases. 

Of course, substantial differences were observed between the target profiles. 
For two target profiles with very different results, we have generated 10.000 siblings for each of them. The results are visualized in Figure \ref{ExtremeRank}. These profiles are the same ones that give rise to the first two figures of Figure \ref{PI-alpha-3prof}. The profile corresponding to the upper graph in Figure \ref{ExtremeRank} is, for a SGMPlus profile, rare (its random match probability being $10^{-17.8}$) whereas the other profile is quite common (its random match probability being equal to $10^{-11.4}$).

\begin{figure}[h]
\caption{Probability for a sibling to obtain SI-rank at most $n$, in Dutch database: two extreme cases and the average}
\label{ExtremeRank}
\includegraphics[width=12cm]{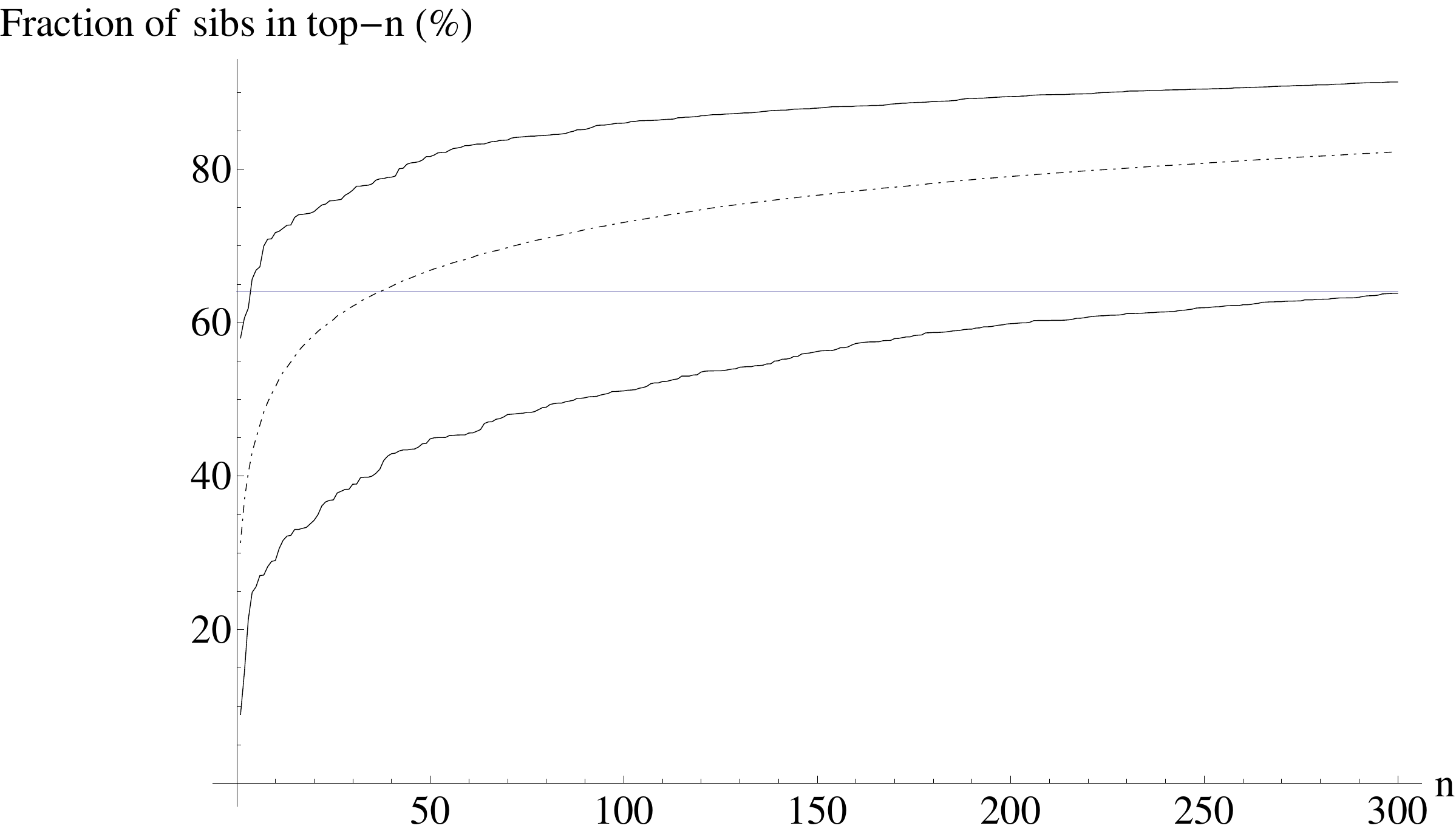}
\end{figure}
Notice that, whereas for the profile corresponding to the upper graph, the probability for a sibling to rank in the top-5 is 67\%, whereas for the profile corresponding to the middle graph the probability to rank in the top-300 is slightly less, 64\%. This illustrates the dramatic differences between different familial searches when it comes to the ease of retrieving a sibling from the database. The middle graph is the average also displayed in Figure \ref{AvRank}.

\subsection{Half-siblings}\label{halfsibs}

Finally we investigate exactly how hard it is to find half-siblings in the Dutch DNA database. It is well known that, using independent autosomal markers as we do here, it is impossible to differentiate between half-siblings, grandparent-grandchild, and uncle-nephew relationships (at least, in the absence of mutation but in practice also when mutation is taken into account with realistic mutation rates). Thus the discussion holds in fact for all these types of relatives, but we will speak of half-siblings only.

Since it is to be expected that half-siblings that genetically look sufficiently like full siblings will be found when a familial search for full siblings is performed, the question is how many half-siblings are found in such a familial search for full siblings. 

We created an artificial database $\D$ with 100.000 SGMPlus profiles (10 loci), and considered 100 artificial target profiles. For each of these, we added 500 times a half-sibling to the database and then sorted to extended database according to decreasing SI or HSI with the target.

The distribution of the ranks that the half-siblings obtained when the database was sorted according to decreasing SI or according to decreasing HSI had the behaviour displayed in Figure \ref{rankHS}.
As the graph indicates, half-siblings are more efficiently found with a HSI-based ranking than with a SI-based ranking, but nevertheless the SI-list does not perform very poorly.
\begin{figure}[h]
\caption{Probability (\%) for a half-sibling to obtain SI-rank and HSI-rank at most $n$, in SGMPlus database, $N=100.000$}
\label{rankHS}
\includegraphics[width=12cm]{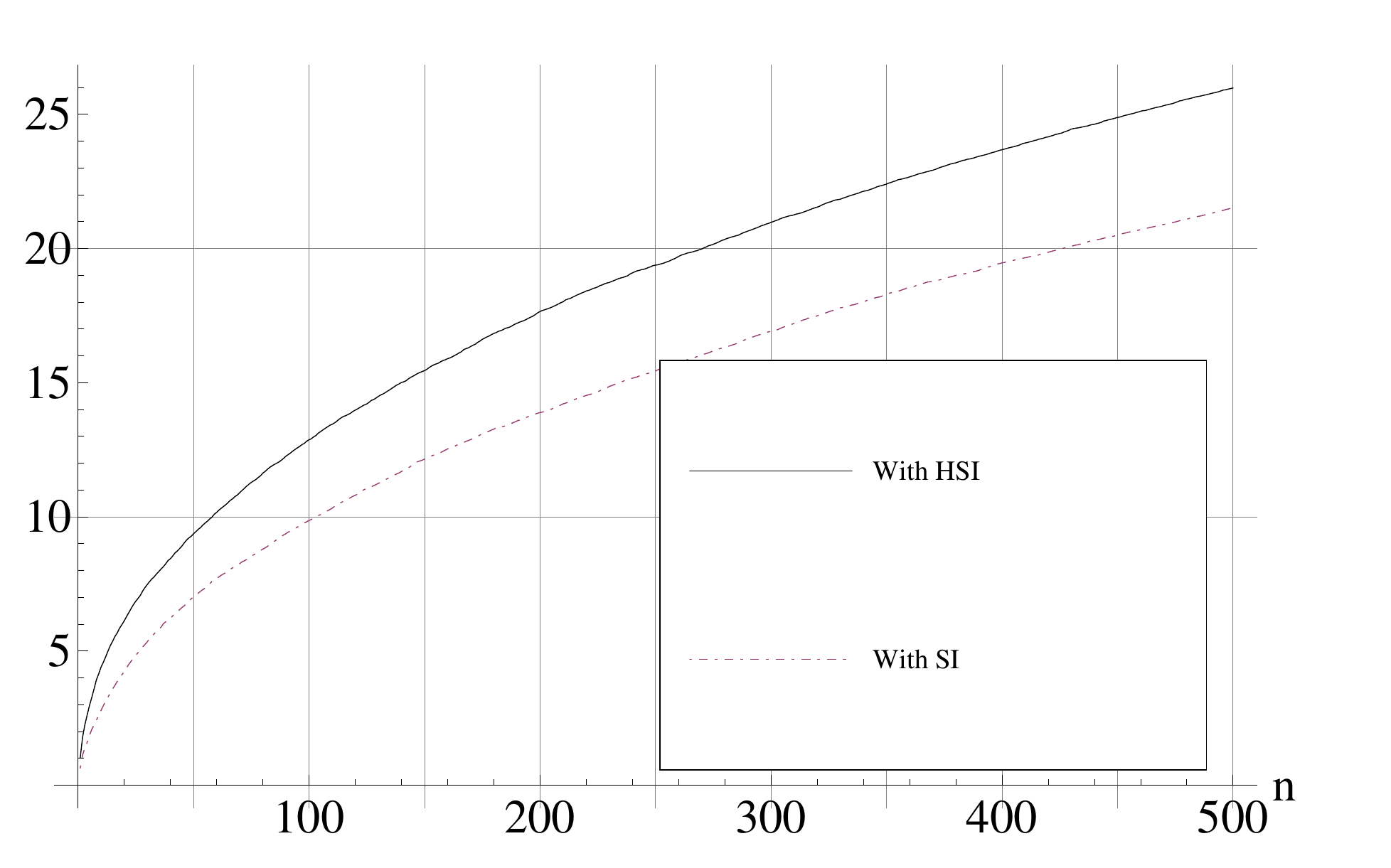}
\end{figure}

We also performed a simulation in which half-siblings were added to a database consisting of 100.000 artificial 15-locus NGM profiles. The result, which is similar, is displayed in Figure \ref{rankHS15} below. We conclude from this that, if the database contains a half-sibling but is only searched for siblings, then there is a quite non-negligible probability that this half-sibling will be found in a sufficiently high position in the SI-ranked list for it to be detectable. When additional DNA tests are done, one should be aware of this possibility; the fact that target and database member have different Y-STR profiles may then be used to infer that they cannot be full siblings, but it must be taken into account that they may nonetheless still be related as half-siblings, uncle-nephew, or (perhaps less likely in a DNA database) grandparent-grandchild.

\begin{figure}[h]
\caption{Probability for a half-sibling to obtain SI-rank and HSI-rank at most $n$, in a database with 100.000 15-locus profiles}
\label{rankHS15}
\includegraphics[width=12cm]{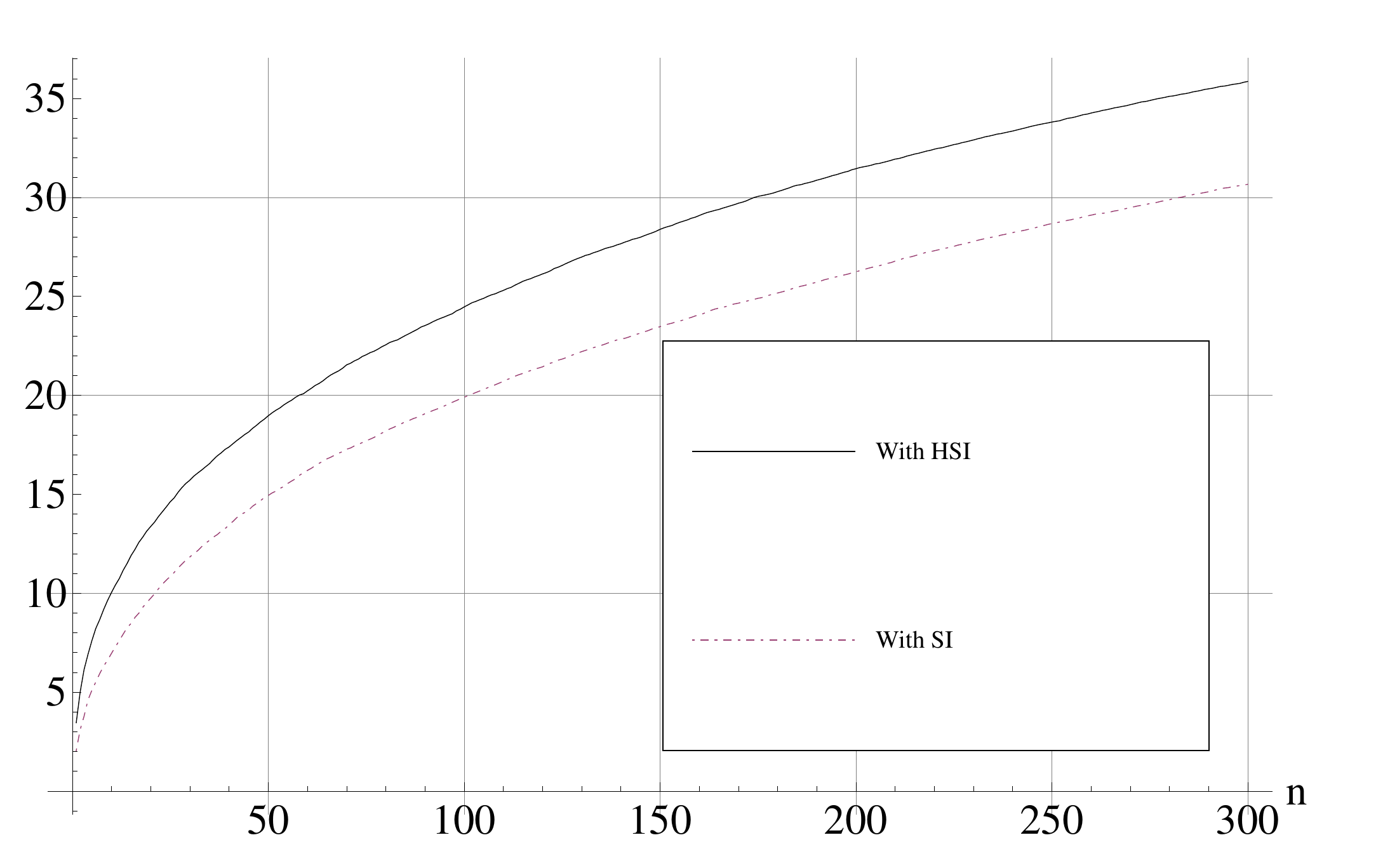}
\end{figure}

\section{Discussion}
We have defined and investigated several search strategies that construct a subset from a database with DNA profiles in such a way that the probability that they contain the relative of the target, if present in the database, with a certain minimum probability. We end this article with a brief recapitulation of their main properties and how this compares to other research published on familial searching.
First, there is the option to use or not to use prior probabilities. Using prior probabilities has the advantage of also obtaining posterior probabilities, as specified in Proposition \ref{postpi}.
It turns out that ranking according to posterior probability is the same as ranking according to the a priori probability multiplied by the likelihood ratio. Hence, in the case of a uniform prior, the database can simply be ranked according to likelihood ratio. This is, of course, well known. What can also be derived from this article however, is how much genetic similarity one expects by chance. Indeed, whatever the database's size, the type of relative, or the number of loci included in the comparison, statistically one expects the sum of the likelihood ratios with all database members to be equal to the number of people in the database, if it does not contain a relative. This is also true for direct matches, but in that case databases are too small compared to the obtainable likelihood ratios to see this effect. In contrast, when looking for relatives, our simulation study shows that for a database with 100,000 members and ten investigated loci, the observed total likelihood ratio agrees well on average with this prediction. In casework, this can give a practical check: if the sum of likelihood ratios is of comparable magnitude as the database size, there is no strong indication for the presence of the target's relative. It does not matter if this total is obtained by one large likelihood ratio or by several smaller ones.

If one does not use prior probabilities, then it is still possible to keep the probability with which the relative is detected under control by using the fact that the likelihood ratios between the target and his relatives have a distribution that can be derived from the target's DNA profile. This means that a lab that carries out familial searching can choose to fix either the desired probability of detection, the minimum likelihood ratio that warrants additional testing, or the number of additionally tested individuals. This also is of course well known, but with the results presented here one is able to derive the unfixed two parameters from the fixed one. For example, a lab may decide to guarantee a certain probability of detection. If it does so, the likelihood ratio threshold for additional DNA testing as well as the number of additionally tested individuals will vary from case to case. They can be computed or estimated before a search is carried out, which can be helpful in the phase where the familial search details are discussed with the investigating authorities, as it will give a clear picture of what to expect. On the other hand, a lab may decide to additionally test a fixed number of individuals in every familial search. If it does so, then the results presented here allow to say what the probability is, in a specific case, that a relative ranks sufficiently high.

We have also discussed the somewhat subtle relation between these two approaches. They are different, yet can have the same efficiency. The reason behind this is that the target-centered method uses other information: admission of a database member into $\D_\alpha$ is on the basis of the likelihood ratio of that database member alone, whereas for admission into $\D^\alpha$ the product of that likelihood ratio with the a priori probability has to be sufficiently big, in comparison to that of the others. The probabilities for the relative to be found by both strategies have a different frequentist interpretation: the conditional method uses more information but needs resampled databases for its frequentist interpretation, and the target-centered method uses less information but does not need resampled databases.

Finally, a good question is if one method seems preferable over another. There does not seem to be, in our opinion, a definitive argument to prefer any of the methods over the other one. As just pointed out, the conditional method has the advantage of being able to deal with prior probabilities. Hence, if these are definable, it makes sense to prefer this method. On the other hand, its probabilistic interpretation is more complex. If no prior odds exist, then the target-centered method has the advantage of having naturally interpretable probabilities, and the expected outcome of a search can be simulated beforehand, since it only depends on the size of the database and the amount of genetic information that it contains. It is therefore very suitable to perform a case pre-assessment: one can obtain estimates of how many individuals must be additionally tested, if one looks for a certain type of relative with a certain probability. Indeed, all that is needed is the target profile, so no comparisons with the actual database are required. Prior to having obtained permission for a familial search, the target-centered method therefore allows a feasibility study.
Once permission has been obtained and the likelihood ratios with the database have been computed, one can then use the conditional method to take these into account. This allows, for example, to make statements about the probability that a relative is found in the particular top-$n$ at hand, as opposed to what would a priori be expected, e.g. based on the results of sections \ref{ranking} and \ref{halfsibs}.

\appendix
\section{}
In this appendix we provide the proof of Proposition \ref{postpi}. Before doing so, we first state some properties of general likelihood ratios which are probably well known, but since we are unaware of an explicit reference, we include them with proofs as well.

\begin{prop}
\label{rec} Suppose that for all $e \in E$ we have
\begin{equation}
\label{basic}
P(\rv{S}=e)>0 \Rightarrow P(\rv{G}=e)>0. 
\end{equation}
Then we have, for all $x\geq 0$,
\[ \frac{P(\rv{LR(S)}=x)}{P(\rv{LR(G)}=x)}=x.\]
\end{prop}

\begin{proof}
Denote the part of $E$ on which the likelihood ratio takes value $x$ by
\[ E_x=\{ e \in E \mid \rv{LR}(e)=x\}.\]
We write
\begin{eqnarray*} P(\rv{LR(S)}=x) &=& \sum_{e \in E_x} P(\rv{S}=e) \\ &=& \sum_{e \in E_x} \rv{LR}(e)P(\rv{G}=e) \\ &=& x \sum_{e \in E_x}P(\rv{G}=e) \\ &=& x P(\rv{LR(G)}=x).\end{eqnarray*}
\end{proof}

\begin{prop}
\label{e1}
Under assumption \eqref{basic} we have $E(\rv{\rv{LR}(\rv{G})})=1$.
\end{prop}

\begin{proof}
We write
\begin{eqnarray*}
E(\rv{LR(G)}) &=& \sum_{e \in E} \rv{LR}(e) P(\rv{G}=e)\\
&=& \sum_{e \in E} P(\rv{S}=e) =1.
\end{eqnarray*}
\end{proof}

Proposition \ref{e1} can be interpreted as expressing that for every choice of likelihood ratio, there will always be chance matches (likelihood ratios in favour of $\rv{S}$ whereas the data were generated by $\rv{G}$), as long as \eqref{basic} holds.
Moreover, if we expect fewer chance matches, then these matches will be stronger to the effect that the expected likelihood ratio is constant.

\medskip\noindent
{\em Proof of Proposition \ref{postpi}.}
Recall that the $\rv{P}_i$ are independent random variables distributed either as $\rv{S}$ or as $\rv{G}$, with exactly one of them distributed as $\rv{S}$. The database consists of individuals $1,\dots,N$.
We write $\rv{D}=(\rv{P_1},\dots,\rv{P_N})$ for the corresponding random vector, and
\[ \rv{LR_D}=(\rv{LR}(\rv{P_1}),\dots,\rv{LR}(\rv{P_N})),\]
for the random vector representing the likelihood ratios that we obtain from the database. 

With this notation, we first prove (\ref{nummereen}). The required probability is equal to
\[ \frac{P(\rv{LR_D}=\r \mid \rv{R}=i)P(\rv{R}=i)}{P(\rv{LR_D}=\r)},\]
which can be expanded as 
\[ \frac{P(\rv{LR_D}=\r \mid \rv{R}=i)P(\rv{R}=i)}{\sum_{k=1}^NP(\rv{LR_D}=\r \mid \rv{R}=k)P(\rv{R}=k)+P(\rv{LR_D}=\r \mid \rv{R}\notin\D)P(\rv{R} \notin \D) }.\]
Therefore, it is more attractive to consider the reciprocal, and we obtain 
\begin{eqnarray*} \frac{1}{P(\rv{R}=i \mid \rv{LR_D}=\r)}&=& \sum_{k=1}^N\frac{P(\rv{LR_D}=\r \mid \rv{R}=k)}{P(\rv{LR_D}=\r \mid \rv{R}=i)}\frac{P(\rv{R}=k)}{P(\rv{R}=i)}\\&+& \frac{P(\rv{LR_D}=\r \mid \rv{R}\notin\D)}{P(\rv{LR_D}=\r \mid \rv{R}=i)}\frac{P(\rv{R} \notin \D)}{P(\rv{R}=i)}.\end{eqnarray*}
Recall that all $\rv{P_i}$ are conditionally independent given $\rv{R}$. This means that the last expression reduces to
\begin{eqnarray*}
& &\sum_{k=1}^N\frac{P(\rv{LR}(\rv{P_i})=r_i \mid \rv{R}=k)}{P(\rv{LR}(\rv{P_i})=r_i \mid \rv{R}=i)}\frac{P(\rv{LR}(\rv{P_k})=r_k \mid \rv{R}=k)}{P(\rv{LR}(\rv{P_k})=r_k \mid \rv{R}=i)}\frac{P(\rv{R}=k)}{P(\rv{R}=i)}+\\
&+& \frac{P(\rv{LR}(\rv{P_i})=r_i \mid \rv{R}\notin\D)}{P(\rv{LR}(\rv{P_i})=r_i \mid \rv{R}=i)}\frac{P(\rv{R} 
\notin \D)}{P(\rv{R}=i)}.
\end{eqnarray*}
We claim that for $k \neq i$, we have
\begin{equation}
\label{eerste}
\frac{P(\rv{LR}(\rv{P_i})=r_i \mid \rv{R}=k)}{P(\rv{LR}(\rv{P_i})=r_i \mid \rv{R}=i)}=\frac{1}{r_i}
\end{equation}
and
\begin{equation}
\label{tweede} 
\frac{P(\rv{LR}(\rv{P_k})=r_k \mid \rv{R}=k)}{P(\rv{LR}(\rv{P_k})=r_k \mid \rv{R}=i)}=r_k.
\end{equation}
To see this, note that given $\rv{R}=k$, $\rv{P_k}$ is distributed as $\rv{S}$, and $\rv{P_i}$ 
is distributed as $\rv{G}$, and then use Proposition \ref{rec}. 
For $k=i$, the corresponding term in the sum is equal to 1. We can also apply Proposition \ref{rec} to the last term since $\rv{R} \notin \D$ implies that $\rv{P_i}$ is distributed as $\rv{G}$. From all this it follows that
\[ \frac{1}{P(\rv{R}=i \mid \rv{LR_D}=\r)}=\sum_{k=1}^N\frac{r_k}{r_i}\frac{P(\rv{R}=k)}{P(\rv{R}=i)}+\frac{1}{r_i}\frac{P(\rv{R} \notin \rv{D})}{P(\rv{R}=i)},\]
and (\ref{nummereen}) follows.

The proof of (\ref{nummertwee}) is similar, we only sketch the difference with the proof of (\ref{nummereen}). The probability in question is equal to
\[ \frac{P(\rv{LR_D}=\r \mid \rv{R}=i)P(\rv{R}=i)}{P(\rv{LR_D}=\r, \rv{R} \in \D)},\]
which can be expanded as 
\[ \frac{P(\rv{LR_D}=\r \mid \rv{R}=i)P(\rv{R}=i)}{\sum_{k=1}^NP(\rv{LR_D}=\r \mid \rv{R}=k)P(\rv{R}=k)}.\] From this point on, the proof proceeds as above.

\bibliography{biblionfi}
\bibliographystyle{amsplain}

\end{document}